\documentclass[preprint,12pt]{elsarticle}
\usepackage{hyperref}
\usepackage[dvipsnames]{xcolor}
\usepackage{amsmath,amssymb} %new
\usepackage[caption=false]{subfig}
\newcommand{\uproman}[1]{\uppercase\expandafter{\romannumeral#1}}
\usepackage{tikz}
\usetikzlibrary {arrows.meta,bending}

\newtheorem{theorem}{Theorem}

\newenvironment{proof}{\emph{Proof.}}{\hfill \qed}%{\hfill Q.E.D.}
\newcommand{\N}{\mathrm{I\kern-.23emI\kern-.29em N}}

\usepackage{amssymb}
\begin{document}

\begin{frontmatter}


\title{Simple grid polygon online exploration revisited}


\author{Maximilian Brock, Martin Brückmann, Elmar Langetepe, Raphael Wude}

\affiliation{organization={University of Bonn, Institute of Computer Science V},
%Department and Organization
%         addressline={},
            city={Bonn},
            postcode={D-53113},
%         state={},
%          country={}
  }

\begin{abstract}
Due to some significantly contradicting research results, we reconsider the problem of the online exploration of a simple grid cell environment.
 In this model an agent attains local information about the direct four-neigbourship of a current grid cell  and can also successively build a map of all detected cells. Beginning from a starting cell at the boundary of the environment, the agent has to visit any cell of the grid environment and finally has to return to its starting position. The performance of an online strategy is given by competitive analysis. We compare the number of overall cell visits (number of steps) of an online strategy to the number of such visits in the optimal offline solution under full information of the environment in advance. The corresponding worst-case ratio gives the competitive ratio.

 The aforementioned contradiction among two publications turns out to be as follows: There is a journal publication that claims to present an optimal
 competitive strategy with ratio~$\frac{7}{6}$ and a former conference paper that presents a lower bound of~$\frac{20}{17}$.
 In this note we extract the flaw in the upper bound and also
 present a new slightly improved and (as we think) simplified general lower bound of~$\frac{13}{11}$.
\end{abstract}



\end{frontmatter}



\section{Introduction}\label{sect-intro}
Back in 2005, Icking et al.~\cite{ltepe} elaborated on the simple grid polygon exploration problem. In this online exploration setting a robot explores an unknown grid polygon without holes while having a limited sensing capability: It can only perceive the four direct neighbouring grid cells. The optimization approach aims to find the shortest exploration path from the starting cell, i.e., the sequence of steps exploring all grid cells and returning to its start position. It is assumed that the agent starts at  the boundary of the polygon. %This restriction can be  relaxed as we will briefly discuss in Section~\ref{sect-model}.
The performance of the online exploration strategy can be measured against the optimal shortest path in the offline situation in which the robot has full knowledge of the grid topology. This ratio is known as competitive ratio and states the performance of an online strategy over all possible polygons relative to the offline optimum. The authors proved that their strategy "SmartDFS" has a $\frac{4}{3}$-competitive ratio. At the same time they presented the first lower bound construction of $\frac{7}{6}$, meaning that against any strategy a simple grid polygon can be presented such that the exploration path in the online setting is at best $\frac{7}{6}$ times the optimal offline path in which case the environment is known in advance. Note that in the case that the agent need not return to the start, there is a lower bound  construction of~$2$ for the competitive ratio and a DFS walk (for the cells) will precisely attain this ratio against any optimum. This lower bound can already be adapted from a more general construction by Icking et al. \cite{icking2002competitive}, compare also Section~\ref{sect-relwork}.

Later in 2010 Kolenderska et al.~\cite{kole} presented an improved $\frac{5}{4}$-competitive strategy as well as an improved lower bound of $\frac{20}{17}$. This further closed the gap between the two bounds. Recently, in 2021 Wei et al.~\cite{wei} proposed an even better exploration strategy with a $\frac{7}{6}$-competitive ratio. Regarding the first lower bound in~\cite{ltepe} they claim to present an overall optimal strategy. Looking at the two latest publications, we have to point out that the previously mentioned $\frac{7}{6}$-competitive exploration strategy (or upper bound) opposes the former lower bound of $\frac{20}{17}$. These results motivated us to resolve the contradiction. After revisiting the problem, we furthermore present an improved and simplified lower bound of $\frac{13}{11}$.

Due to the recent upcoming flaws we decided to recapitulate the model and the corresponding  facts very precisely. The details turn out to be very important and this write-up gives a somewhat self-contained overview.

\section{Notation and model of simple grid polygons}\label{sect-model}
We consider a regular square grid environment where any grid cell (square) is either a \emph{free} or a \emph{boundary} cell. At any point in time the agent is located in one of the free cells and can check
 the status of the four direct neighbouring cells. It is possible to move into a free neighbouring cell by one step.  We assume that the agent starts in a free cell~$s$ that is adjacent to at least one boundary cell. The task is to explore (visit) all free cells in the (cell neighbour-ship) connected component (of free cells) of~$s$ and return to~$s$. So we assume that the connected component of~$s$ is bounded or surrounded by (a finite number of) boundary cells. A movement from one cell to a free neighbouring cell is counted as a single step. We do not take turning costs into account. The agent can successively build a map of the environment.
 We further assume that the connected component of free cells is \emph{simply} connected which is interpreted as such there are no \emph{holes} inside. I.e., a set of boundary cells \emph{fully} surrounded by free cells is forbidden. 
  Therefore we can also interpret the path along the outer boundary of the cell-environment as a simple polygon without self-intersections
 (but may be with touching); compare Figure~\ref{fig-mod} for illustration. Inside the polygon (boundary) there are only free cells. For example in Figure~\ref{fig-mod}(iv) the boundary cell~$B$ is not a hole, it is not fully surrounded by free cells and does not lie inside the polygon.
 Grid cell environments of these kind are therefore denoted as \emph{simple grid polygons}. 
 In Section~\ref{sect-relwork} we will briefly compare simple grid polygons to the notion of \emph{solid grid graphs}, they are almost the but not exactly the same. 

  \begin{figure}[htp]
 \begin{center}
  \includegraphics[scale=0.6]{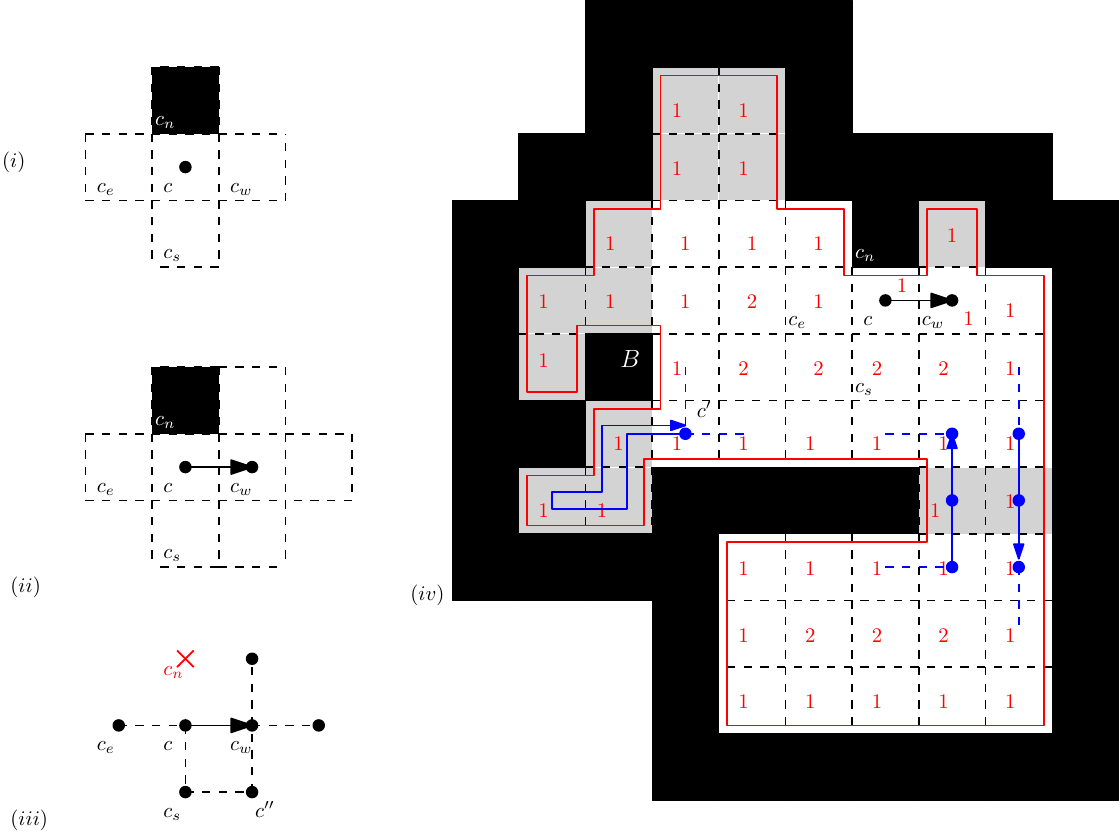}
\caption{(i) The local view of the agent in the beginning, only the status (free/boundary) of the direct 4 neighbouring cells can be achieved.
(ii) The currently known cell map after one step from $c$ to the neighbouring cell. (iii) The model can also be directly interpreted as a (special) grid graph model where cells are represented by vertices located in the center of the cell. Note that by visiting a \emph{vertex} $c_w$ and detecting the edge $(c_w,c'')$ the edge $(c_s,c'')$ has to exist for a one-to-one correspondence. In the given model it is also already known at position $c$ that a connected vertex $c_n$  cannot exist, this can be different in the notion of \emph{solid grid graphs}.  (iv) The aforementioned local situation embedded in a simple grid cell environment also denoted as a \emph{simple grid (cell) polygon} by interpretation of the boundary path as a simple polygon (in red). Here the boundary cell~$B$ is not a hole, it is not fully surrounded.
The grey cells indicate the five independent narrow passages of the environment. By (online) DFS movements keeping along the boundary these passages will always be passed optimally (exemplified by some blue paths). We do not count turnings, so the blue (sub-)path in the left/lower narrow passage starting and ending at $c'$ consists of 6 steps.
}
\label{fig-mod}
\end{center}
\end{figure}

 Free grid cells that touch the boundary (also by a single vertex of the square) belong to the first layer. The second layer can be defined recursively by deleting the cells of the first layer and considering the corresponding new grid polygons (could be more than one by disconnection). A grid cell of the first layer belongs to a \emph{narrow passage} when the
 local deletion of the cell does not change the layer number of its direct adjacent neighbours. A narrow passage is a connected collection of such cells. For example, in Figure~\ref{fig-mod} the
 cell $c_e$ is in the first layer but it is not part of a narrow passage, the deletion of $c_e$ will cause that its neighbour in the south will be in the first layer of the resulting grid polygon.

  If an agent follows the boundary by DFS left-hand (or right-hand) rule, a narrow passage will be traversed optimally even though double visits occur. A narrow passage always has precise entrance and quit cells also located in the first layer. In the following we will omit the full black boundary cells and only use boundary edges of the simple grid polygon.

Note, that a simple full (online) DFS walk on the free cells/vertices results in a full round-trip that visits all reachable cells and uses no more than $2(V-1)$ steps for $V$ cells. 
So this is already a 2-approximation against an optimal offline solution for the same task. In other words DFS is 2-competitive. 
%Note, that a simple (online) DFS walk on the free cells any free cells at most twice
%and gives a competitive ratio of~$2$ 
This is a tight bound, 
for cell environments with holes, it was shown that there is no strategy that attains a better factor than~$2$. So the precise competitive complexity below~$2$for simple grid environments is still of some interest,  see Section~\ref{sect-relwork} for more details.
\bigskip

The notion of \emph{competitive analysis} goes back to Sleator and Tarjan \cite{sleator1985amortized} and was first used in the context of list update and paging problems. Competitive analysis measures the quality of an online algorithm under incomplete information by comparing it against an optimal offline algorithm with complete knowledge. 

Let us consider a specific algorithmic problem and let $\mathcal{P}$ be the set of all instances of the problem. An optimal offline algorithm \emph{OPT} solves the problem and runs under full information. Furthermore, the cost of such an algorithm for an instance  $P\in \mathcal{P}$ is denoted as $S_{OPT}(P)$ and optimality means that \emph{OPT} attains the minimal cost for any instance and among any algorithm. Let us further assume that an online algorithm \emph{ALG} also solves any problem instance of  $P\in \mathcal{P}$ but attains the problem information successively. Thus \emph{ALG} might not solve the problem with minimal cost, i.e.,  $S_{ALG}(P)\geq S_{OPT}(P)$.  We call the online algorithm $C$-competitive if there are constants $A$ and $C$ such that the inequality 
\begin{equation}\label{equ-compAn}
S_{ALG}(P)\leq C \cdot S_{OPT}(P) + A
\end{equation}
holds for all $P\in  \mathcal{P} $. So \emph{ALG} guarantees to be no worse than~$C$ times the optimal offline solution, apart from an additive constant~$A$. This second constant~$A$ covers starting situations. It is required when for some fixed small instances some first steps of the online algorithm cause relatively high cost against the optimum. In this case the problem might not be analysed properly. Thus it is allowed that small instances or special starting situations can be fully covered or excluded by additive constants. 

In the given grid environment situation  for the analysis of concrete exploration strategies it was (mainly) guaranteed that the ratio $\frac{S_{ALG}(P) }{S_{OPT}(P) }\leq C$ directly holds for any $P\in \mathcal{P}$. This means that the additive constants have been neglected or were never used within the analysis of a strategy. In our setting the precondition of starting the exploration on the boundary can be relaxed by making use of an additive constant. For the construction of a lower bound on~$C$, i.e., try to prove that no strategy approximates better than a fixed constant~$C$, a potential additive constant always has to be taken into account by definition.  To this end we guarantee that for any $\epsilon>0$ we can construct instances $P_{\epsilon} \in  \mathcal{P}$ such that $\frac{S_{ALG}(P_{\epsilon}) }{S_{OPT}(P_{\epsilon}) }> (C-\epsilon)$ holds and the construction allows to let such $P_{\epsilon}$ and especially the cost $S_{ALG}(P_{\epsilon})$ become arbitrarily large. Therefore, we will finally overrun any fixed additive constant~$A$ and the ratio $C$ indeed gives a  lower bound on the competitive ratio for any strategy, see Theorem~\ref{theo-LB}.

\section{Related work}\label{sect-relwork}
It is very obvious that the online grid cell exploration problem can be easily  transformed to a (special) online grid graph problem. A free cell in the grid cell environment represents a vertex of the graph and two direct adjacent free cells are connected by an edge; see Figure~\ref{fig-mod}~(ii) and~(iii) for the precise correspondence.
For this reason any work regarding the exploration of the vertices of a planar grid graph is related to our topic.

In the literature the model of the simple cell environment described in the previous section is also closesly related to the notion of a \emph{solid planar grid graph}. A planar grid graph is a (connected) graph $G=(V,E)$ with vertex set $V$ as a (sub)set of the planar integer lattice and with edges $e=(v,w)\in E$ that connect any pair of such vertices that are unit distance away from each other. Therefore such graphs can be considered to be fully induced by its vertices. Such an integer lattice grid graph $G$ is called \emph{solid} if it does not have any \emph{holes}, i.e., its vertex complement in the planar integer lattice is (unit-edge) connected. Keep in mind that this interpretation would mean that the cell~$B$ in Figure~\ref{fig-mod}(iv) indeed counts as a hole, there is no edge for connecting $B$ to the rest of the boundary cells (vertices).   
So \emph{solid planar grid graphs} and \emph{simple grid polygons} are not precisely but almost the same. For the argumentation used in this paper  (lower bound and flaws) there will be no difference because this situation does not occur. In general one has to take care of the slightly  different interpretations and models.   

The computational complexity of computing the (offline) optimal Travelling Salesman Route for a solid grid graph is (to our knowledge) still unknown. It still belongs to a collection of open geometric problems; Problem 54 at the \href{https://topp.openproblem.net/p54}{\emph{Open Problem Page}}.

The following facts for the offline (full-information) and the online (partial information) variants are known. For the online versions we assume a local detection of outgoing edges and the construction of a partially (known) map during the movements. 

In the offline case for general planar grid graphs the Hamiltonian-Cycle problem and the Traveling-Salesman problem are
NP-complete by Itai et al. \cite{itai1982hamilton}.
For solid planar grids it is possible to decide in polynomial time whether a Hamiltonian-Cycle exists and in this case the cycle can also be computed as shown by Umans and Lenhart~\cite{umans1997hamiltonian}. This does not help to find the optimal path when the answer for the cycle is "No", Problem 54 remains open. 

For the online version a simple online DFS strategy on the vertices of a general connected planar grid graph requires at most 2~times more steps than an optimal offline strategy, this  holds already for any graph in the given online model. It was shown by Icking et al. \cite{icking2002competitive} that on the other hand there is no such online strategy that can be better than a factor of~$2$ against the number of steps (vertex visits) of an optimal exploration strategy. This lower bound construction also works for planar grid graphs but it makes use of a (single and rectangular) hole. 

Furthermore, this lower bound construction can be easily extended to the case that the agent need not return to the start. In this setting, there is no strategy that can attain a better competitive ratio of~$2$ and DFS on the cells(vertices) precisely attains this ratio. Thus the competitive complexity of this case is precisely known. The lower bound can simply be given by a  horizontal cell-corridor of width~$1$ and unknown extension to the left and the right.

The just mentioned general lower bounds of~$2$ and the performance guarantee of~$2$ by a simple DFS walk on the cells, might indicate that the online exploration story is mainly over for general grids. 
But it is possible to refine the analysis by being more case sensitive w.r.t. the,  say, the overall \emph{fleshyness} of the environment.
It was shown by Gabriely and Rimon~\cite{gabriely2003competitive} that there is a strategy on general grid graphs that requires at most
$|V|$ + $|B|$ exploration steps, where $|V|$ denotes the number of vertices and $|B|$ denotes the number of vertices of the environment which are adjacent to the boundary  (here also diagonal adjacency counts). In the worst-case any vertex of the environment is adjacent to the boundary, and this is also given in the overall lower bound construction mentioned above from \cite{icking2002competitive}.

Also for the offline case some more progress has been done for finding the  optimal shortest exploration path in grid graphs, some kind of taxonomy was given by Arkin et al.~\cite{arkin2009not}.

One of the most challenging remaining interesting configurations
are solid grid graphs in the online and also in the offline version.
Some attempts in the offline case were done by Fekete et al.~\cite{feketetraveling}. For the online case the current best
known upper bound is $\frac{5}{4}$ in \citep{kole}, and the best lower bound of~$\frac{13}{11}$ will be presented here.

The problem for the analysis of a given strategy for solid planar grids is somehow intrinsic, because the optimal offline approach is not known in general. Furthermore, any analysis has to cope with a tedious case analysis which could make the analysis fragile as we will see below.

\section{Upper bound and recent flaw}\label{sect-flaw}
Almost any reasonable strategy behaves in the same general way. The agent runs in a DFS fashion with a so called left-hand (or right-hand) preference rule, see for example Figure~\ref{fig-flaw}~(ii). This lets a strategy generally run with left-hand along the boundary or also along already visited cells in order to visit new cells by single steps more or less optimally. If a decomposition of the overall cell environment is detected at a so-called \emph{split-cell}, the agent has to decide which of the components should be visited first.  It is a reasonable idea to prefer the component that is \emph{farther away} from the start, because the task is to finally come back. For example in Figure~\ref{fig-flaw}~(ii) by DFS (left-hand) the split-cell $s_1$ is detected and the component $C_1$ is visited first, recursively.  Note that in case of a decomposition in principle the outer boundary of at least one component (here $C_1$) is fully known but in general there is no polynomial time algorithm for this offline setting yet. Therefore, it is a reasonable idea to start the same strategy recursively in the component.

The current best strategy for the above exploration problem was presented 2021 in the \emph{Journal of Combinatorial Optimization} by Wei, Sun, Tan, Yao and Ren,  see~\cite{wei}. The authors claim to present a~$\frac{7}{6}$ competitive strategy in the above model. We briefly discuss the main new idea of this strategy  and the reason for the flaw in the analysis.

\begin{figure}[htp]
\begin{center}
  \includegraphics[scale=0.6]{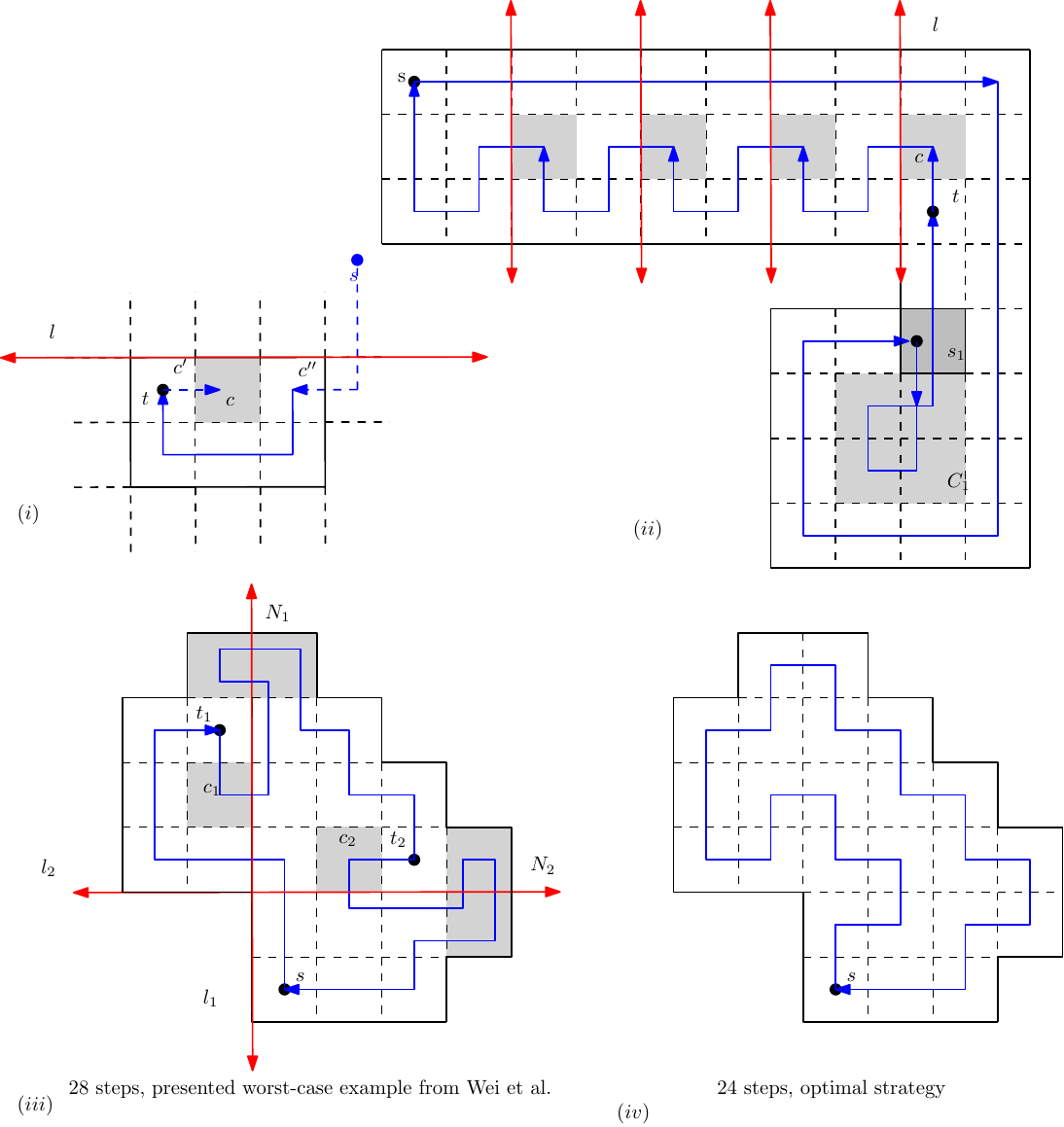}
\caption{(i) The strategy of Wei et. al~\cite{wei} tries to avoid single locally isolated cells $c$.
To this end a tangent line is computed. If it separates the currently active start cell~$s$ from the cell $c$, the cell $c$ is visited first. (ii) The idea is to meander in remaining corridors of width three optimally by applying this rule successively, first at time step~$t$ and then in the whole remaining corridor.
(iii)~The \emph{worst-case} example presented in
Figure~8 of Wei et. al~\cite{wei} applies this rule twice, see $l_1$ and $l_2$ for time stamps $t_1$ and $t_2$ and the corresponding \emph{isolated} cells $c_1$ and $c_2$, respectively. The strategy locally changes the preference and then moves on with
left-hand rule. 28 steps are required in total which gives a ratio of $\frac{7}{6}$ against the optimal offline solution which requires 24 steps only as presented in (iv).}
\label{fig-flaw}
\end{center}
\end{figure}

 Apart from the above very general efficient idea, the strategy of Wei et al.~\cite{wei} additionally tries to avoid to leave single already somehow locally isolated or surrounded cells~$c$ non-visited until coming back later.  Note that these cells are not split-cells, a cell-decomposition from the (recursive) start cell is not given.
   This is briefly indicated in Figure \ref{fig-flaw}~(i). If the agent has at some point in time~$t$ surrounded a single cell~$c$ by visiting 5 neighbouring cells around~$c$  in a $U$-shaped form before, the strategy sometimes should first visit~$c$ and then move on with the general schedule, for example the left-hand rule.
   Note that it is not important in what precise order the surrounding cells have been visited before.  Technically this new
 behavior is implemented by building a tangent line $l$ along the
 free cell~$c$ and its already visited (horizontal or vertical) neighbours~$c'$ and~$c''$. In the given situation the agent is currently (at time $t$) located in one of these direct neighbours. The tangent~$l$ additionally splits the currently known cell environment in different components and if now the (recursively defined) currently active start cell~$s$ is separated from the sequence $c',c,c''$, then the agent should first directly visit~$c$ before moving on as usual.

 The idea of this local behaviour is to meander in corridors of width three optimally, see Figure~\ref{fig-flaw}~(ii).
The authors claim that for this local behavior and the overall strategy there is a  worst-case example for the competitive ratio
of $\frac{7}{6}$. This example was already presented (with the given worst-case attribute) in Figure~8 of Wei et al.~\cite{wei}. In this example the above simple rule is applied twice as presented in our redrawing in Figure~\ref{fig-flaw}~(iii). By the online strategy in total four cells are visited twice.
 The optimal offline path visits any cell only once and is given in Figure~\ref{fig-flaw}~(iv). The ratio is exactly~$\frac{28}{24}=\frac{7}{6}$, this is so far correct.

 Unfortunately, the presented small example  definitely cannot be considered to be a finite worst-case for this specific policy.  In Figure~\ref{fig-new-flaw}~(i) we have slightly changed the environment to a more compact example, the specific rule is still again applied twice but the ratio against the optimum given in Figure~\ref{fig-new-flaw}~(ii)  is
 $\frac{26}{22}=\frac{13}{11}>\frac{7}{6}$. 

 \begin{figure}[htp]
 \begin{center}
  \includegraphics[scale=0.6]{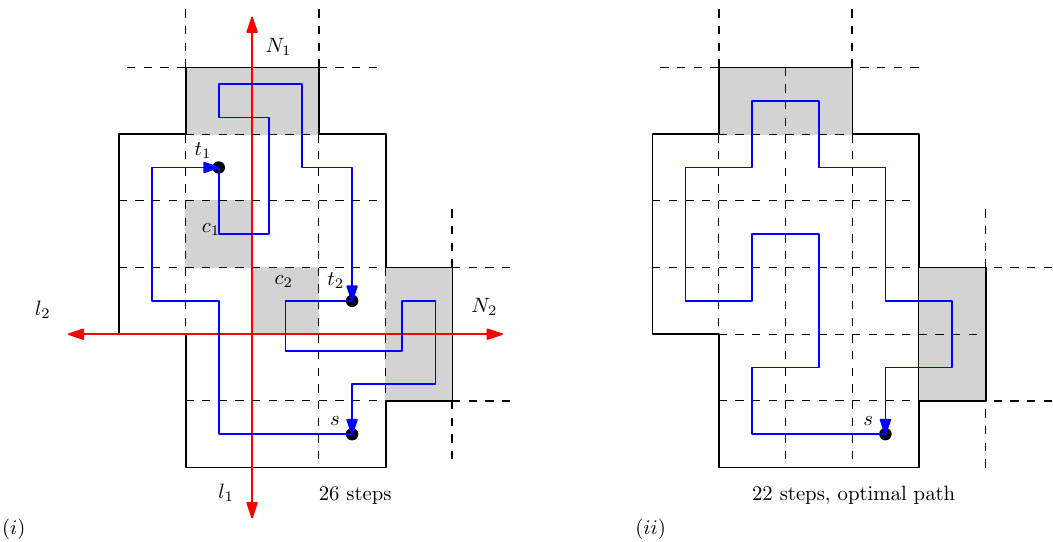}
\caption{The situation can be strengthened such that the ratio of the length of the path of the given strategy (i) and the optimal path length shown in (ii) attains a ratio of $\frac{26}{22}=\frac{13}{11}>\frac{7}{6}$. The narrow passages $N_1$ and $N_2$ both cause~2 extra revisits by running in the passage and also for reentering the remaining part again, respectively. In the optimal solution (ii) narrow passages are resolved optimally along the boundary.}
\label{fig-new-flaw}
\end{center}
\end{figure}

Note that the above given finite and small example can still not be considered as a lower bound on the competitive ratio of the strategy. The detour in a single fixed environment can be fully subsumed by the additive constant.
Lower bound constructions have to come along with arbitrarily large
grid cell polygons. We have to take care that the cost of the strategy and the cost of the optimal strategy can attain arbitrarily large values. The bound will be attained asymptotically.   

So in the first place where is the general flaw in the proof of the paper beyond the given $\frac{13}{11}$ local worst-case? The authors make use of a decomposition idea which was already used by Icking et al. in 2005. Starting with the pure DFS preference rule (left-hand) along the boundary, the so-called \emph{narrow passages} will always be explored optimally and therefore for the analysis it is sufficient to consider environments without narrow passages, so they can be omitted. The paper of Wei et al.~\cite{wei}
reuses this analysis idea and so they totally neglected narrow passages. But by the specific change of the preference (also in the first round) it happens that there can be additional visits  caused by (for example very simple) narrow passages of width~2. This is for instance already given in the presented \emph{new} single worst-case above.  In Figure~\ref{fig-new-flaw} (i) both $N_1$ and $N_2$ are narrow passages and both induce extra visits when this strategy runs into the passage and leaves the passage (and goes back to the rest of the environment). Narrow passages have to be taken into account  for this specific strategy, they cannot be omitted. Note that the presented passages could also have been extended beyond $N_1$ and $N_2$, nevertheless there are additional extra visits after entering and  after coming back from the passage into the sub-region without the narrow passage. This is obviously one core of the flaw and directly depends on the specific rule of preferring encountered cells first rather than directly moving on with the left-hand rule along the boundary.

We will now present a general lower bound construction that can also be applied to the strategy of Wei et al.~\cite{wei}. It shows that no strategy can be better than~$\frac{13}{11}$ competitive.

\section{Lower bound improvement}
In their lower bound construction, Kolenderska et al. \cite{kole} made use of sets of polygons and were able to achieve a lower bound of $\frac{20}{17}$ by merging (instead of concatenating) polygons. Finally,  in total  11~different sub-polygons (blocks) which were further subdivided in 7~different ratio categories have been used. Partially, the sub-polygons had to be merged in different subareas; compare Figure~5 of Kolenderska et al. \cite{kole}.

We apply a similar merging technique but make use of a different smaller set of polygons and categories and can also increase the bound to $\frac{13}{11}$. We reduce the total number of used  sub-polygons to~5 and they precisely represent the ratio categories.
 For further simplification we let the concatenation and merge always appear in exactly the same way (at a narrow passage rectangle of two cells).
Since we simply construct the
concatenation of the blocks in one fixed direction (horizontally) at well-defined entrance-rectangles and make use of a few number of potential blocks, our construction is (as we think) easy to follow and to implement.

  \begin{figure}[htp]
  \begin{center}
  \includegraphics[scale=0.57]{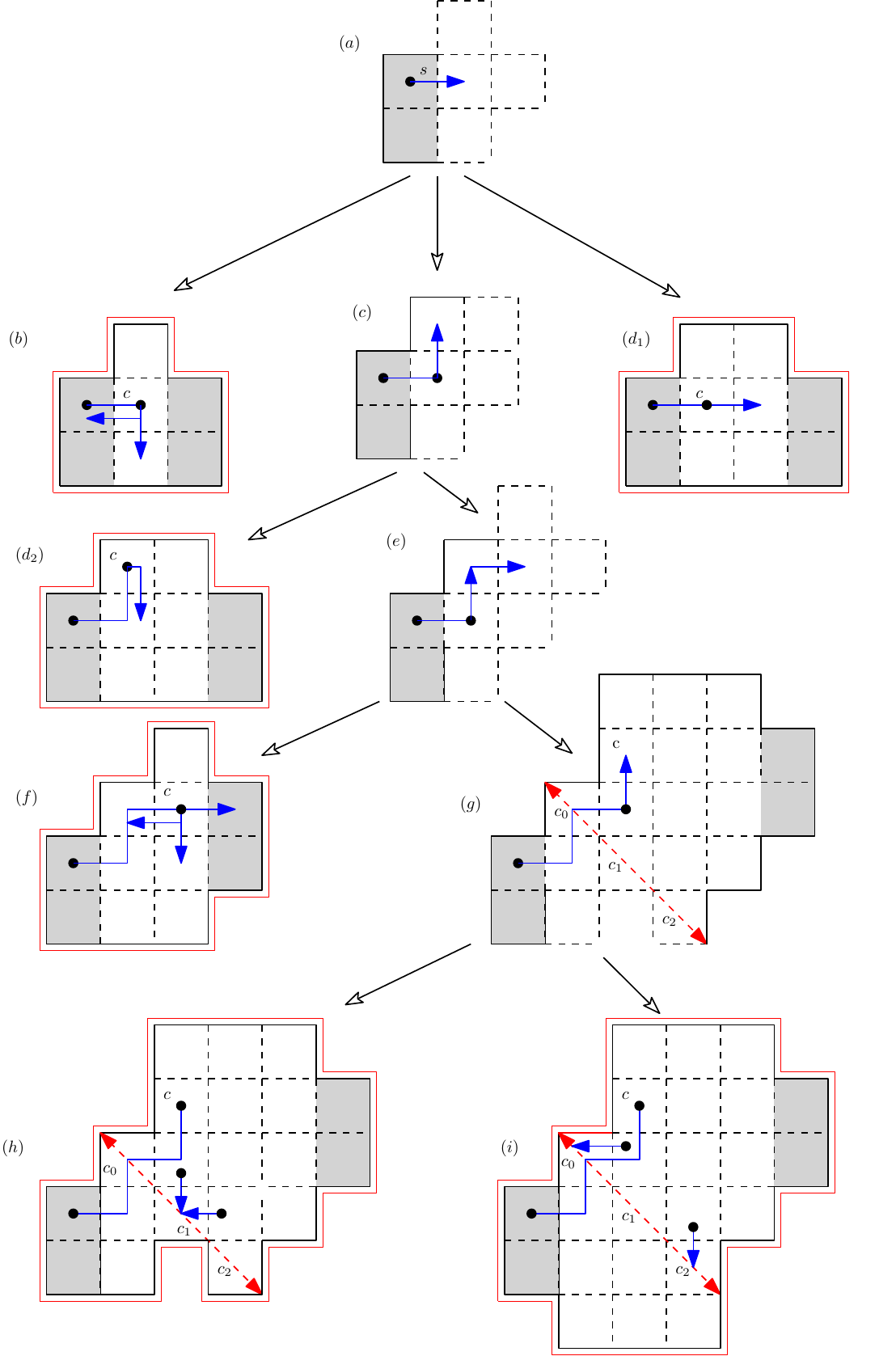}
\caption{The general scheme of the adversary strategy.
The strategy starts in a (potentially already known) rectangle of height~2 and width~1. The adversary strategy starts when this rectangle is left to the right.
For any strategy the adversary finally presents one of the blocks (b), (d), (f), (h) or (i), depending on the motion of the agent. In these final blocks each blue arrow indicates a different option for the movement starting from cell $c$, for different blue arrow options the same block will be presented.
The finalization of (g) in  block (h) or (i) depends on the first visit (blue arrow options) of a cell $c_0$, $c_1$ or $c_2$ after~$c$.} 
\label{fig-lb-total}
\end{center}
\end{figure}

The lower bound is shown in four steps.

\begin{enumerate}
\item Recursively the starting situation is a rectangle of height~2 and width~1 (a narrow passage) that has to be left to the right. A block adversary scheme is presented that finally ends in single fixed blocks. Figure~\ref{fig-lb-total} presents the general scheme and Figure \ref{fig-start-sit}
shows different starting situations that can be handled analogously.
 Any block offers a potentially new starting rectangle on its right side.  
\item A comparison of the optimal strategy against any strategy for these blocks separately as indicated and exemplified in Figure \ref{fig-lb-total-compare} and also given in Table~\ref{table-compare}.
\item A successive horizontal concatenation and merge of such blocks at the corresponding starting and ending narrow passages for the final analysis. For~$n$ successive blocks by this successive merge the online strategy can profit. But this is restricted to at most $2(n-1)$ steps less in total and the optimal strategy also precisely has this gain.
\item Finally, we collect all this information and these arguments for the statement of Theorem~\ref{theo-LB}. An online strategy that tries to avoid a maximal detour for an arbitrary long sequence of horizontal blocks has to trigger block (i), successively. Otherwise the ratio will be even worse than $\frac{13}{11}$.
\end{enumerate}
\subsection{Starting situation and adversary scheme}
We first consider the general single block adversary scheme (see Figure~\ref{fig-lb-total}). The agent starts in a (potentially known) horizontal starting rectangle of height~2 and width~1. We present the scheme for starting in the upper cell and moving directly to the right. I.e. left-hand rule is applied in the first step and we start in the upper cell of the starting rectangle. Note that the same scheme can be applied for
all possible four starting situations in a (potentially known) starting rectangle shown in Figure \ref{fig-start-sit}. If the rectangle is left from the lower cell indicated by $(a_2)$, we simply can mirror all blocks of Figure~\ref{fig-lb-total} horizontally. Thus the given scheme covers all options for the start and for all successive horizontal presentation of the blocks.

\begin{figure}[htp]
\begin{center}
  \includegraphics[scale=0.6]{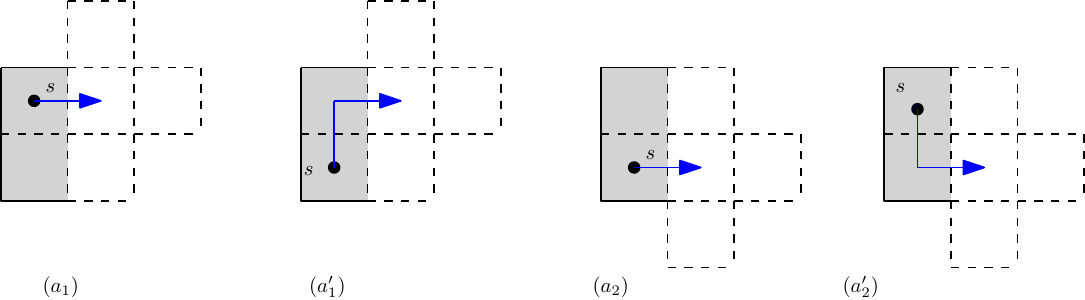}
\caption{All possible starting situations when the agent leaves a starting rectangle of height~2 and width~1 to the right. After that the adversary makes use of the scheme presented in Figure~\ref{fig-lb-total}.
For $(a_2)$ and $(a_2')$ the scheme of Figure~\ref{fig-lb-total} simply has to be mirrored horizontally.}
\label{fig-start-sit}
\end{center}
\end{figure}

\begin{figure}[htp]
 \begin{center}
  \includegraphics[scale=0.58]{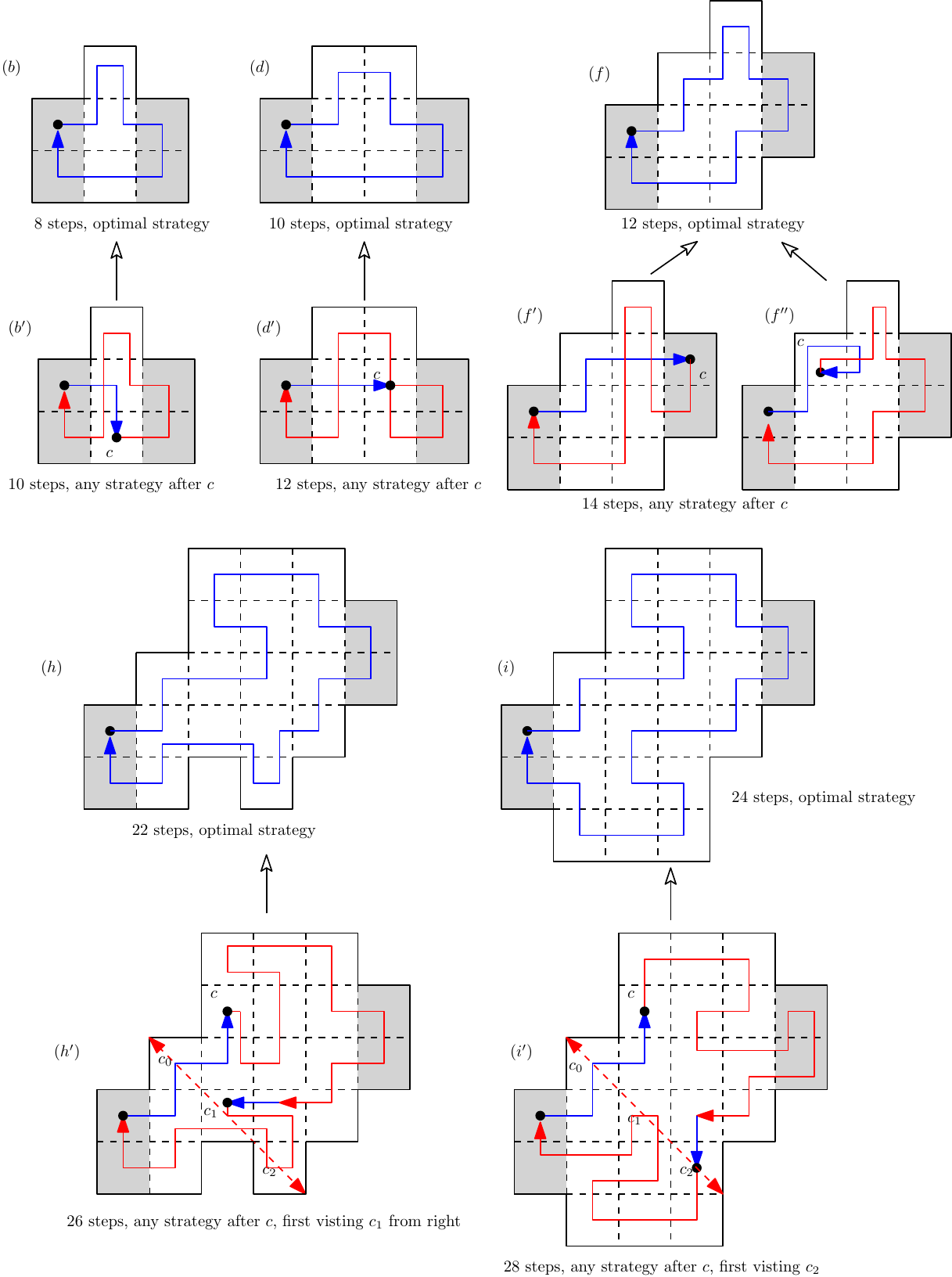}
\caption{Comparison of some online strategies movements to the optimal strategy for block (b), (d), (f), (h) and  (i). Here these blocks have been triggered by the movements in (b'), (d'), (f'), (f''), (h') and  (i'), respectively. The red paths in (b'), (d'), (f'), (f'') are optimal with respect to finishing the block and moving back to the start, two more steps are required. In (h') and (i')
the red path examples from $c$ to $c_1$ or $c_2$ triggered the final construction of the blocks, respectively. Four more steps against the optimum can never be avoided also for all other optional intermediate paths for (g) in Figure~\ref{fig-lb-total}.
}
\label{fig-lb-total-compare}
\end{center}
\end{figure}

For any strategy the adversary scheme will end in one of the 5 final blocks
(b), ($\mbox{d}_{1/2}$), (f), (h) and (i), depending on the  movement of a strategy before.
After the presentation, the blocks still have to be fully finished. In~(g) after visiting cell~$c$, the environment is only partially presented, some options are still left,
depending on the ongoing first visits of the cells $c_0$, $c_1$ or $c_2$, respectively. One of these cells will be visited first for finishing the exploration in the end. 
Some of the final blocks in Figure~\ref{fig-lb-total} cover the situation for different movements, indicated by the different blue arrows.
For example block (i) is used when the strategy (after having visited $c$) will visit either $c_0$ or $c_2$ first in comparison
to visiting $c_1$. So here two options (first $c_0$ or first $c_2$) are covered.
Block (h) is used when $c_1$ is visited first after starting at $c$, this can happen from two different direction or neighbouring cells, so here also two options are covered.

 At the given point in time (when blocks (b), ($\mbox{d}_{1/2}$), (f), (h) and (i) will be fully presented) the strategy can now  finish the rest of the block optimally. The overall construction has the following intention.  Against the overall optimal path any strategy will have to make at least two more steps (cell visits) than the optimal strategy for (b), ($\mbox{d}_{1/2}$) and (f) and at least four more steps for (h) and (i). Note that for (h) and (i) the intermediate movements from $c$ to the cells $c_0$, $c_1$ or $c_2$ can be very different. Additionally, the construction can be easily extended. Any block offers a well defined new starting rectangle to the right. An overall strategy could thus move further on to the right, an arbitrary number of blocks can be presented by an adversary strategy, block copies of (h) and (i) will be finalized during the backward movement. 

 \begin{table}[htp]
\begin{center}
\setlength{\arrayrulewidth}{0.4mm}
\begin{tabular}{ |p{2.5cm}| p{2.5cm}| }
\hline
& \\[-1ex]
\hskip 0.35cm$P_{(\ell)} \in \mathcal{P}$& \hskip 0.2cm$\frac{S_{ALG}(P_{(\ell)})\rule[-3pt]{0pt}{0pt}}{S_{OPT}(P_{(\ell)})\rule[6pt]{0pt}{0pt}}\geq $\\[+2.3ex]
\hline
& \\[-1ex]
\hskip 0.6cm $P_{(b)}$ & \hskip 0.7cm $\frac{10}{8}$ \\[+1.5ex]
\hskip 0.6cm $P_{(d)}$ & \hskip 0.7cm $\frac{12}{10}$ \\[+1.5ex]
\hskip 0.6cm $P_{(f)}$ & \hskip 0.7cm $\frac{14}{12}$ \\[+1.5ex]
\hskip 0.6cm $P_{(h)}$ & \hskip 0.7cm $\frac{26}{22}$ \\[+1.5ex]
\hskip 0.6cm $P_{(i)}$ & \hskip 0.7cm $\frac{28}{24}$ \\[+1.5ex]
\hline
\end{tabular}
\vspace*{5mm}
\caption{$P_{(l)}\in  \mathcal{P} $ denotes the grid polygon of the block $(l)$ for $l\in\{b,d,f,h,i\}$, the following lower bound on the ratios can be achieved by the corresponding blocks.}\label{table-compare}
\end{center}
\end{table}

 \subsection{Comparison to the optimal strategy}
In a second step we now analyse the ratios separately for each block (see Figure \ref{fig-lb-total-compare} and Table~\ref{table-compare}) against the optimum as intended. In  Figure \ref{fig-lb-total-compare} all optimal solutions for the blocks are given. We present some exemplifying reasonable efficient movements of the strategies after the blocks have been fully (or
partially) presented, these movements are given by red paths. For convenience, we did not present any strategy option (blue arrows) given in Figure~\ref{fig-lb-total}. It is obvious that the remaining options work in the same way.

For the formal comparison let $\mathcal{P} $ denote the set of all simple grid polygons.
For $P\in  \mathcal{P}$ let $S_{ALG}(P)$ be the number of steps for an online strategy $ALG$ for a grid polygon $P$ and let $S_{OPT}(P)$ denote the number of steps for an optimal strategy $OPT$ for $P$. If $P_{(l)}\in  \mathcal{P} $ denotes the
polygon of the block $(l)$ for $l\in\{b,d,f,h,i\}$, we finally can guarantee ratios as shown in Table~\ref{table-compare}  depending on what kind of block the online stragegy $ALG$ will \emph{produce}.

  Note that we will not present any possible movement
(any possible red path) under the given constraints in Figure \ref{fig-lb-total-compare} since there can be very many inefficient movements for a given strategy. Also for (h) and (i) it is almost trivial to see that any strategy  has to make the desired extra (four) visits after starting from $c$, then visiting either $c_0$, $c_1$ or $c_2$ first and then finishing the fully given block. For the partially given starting configuration (g) in Figure~\ref{fig-lb-total} with starting cell $c$ two examples with locally optimal movements are presented in Figure \ref{fig-lb-total-compare}.

\begin{figure}[htp]
\begin{center}
  \includegraphics[scale=0.58]{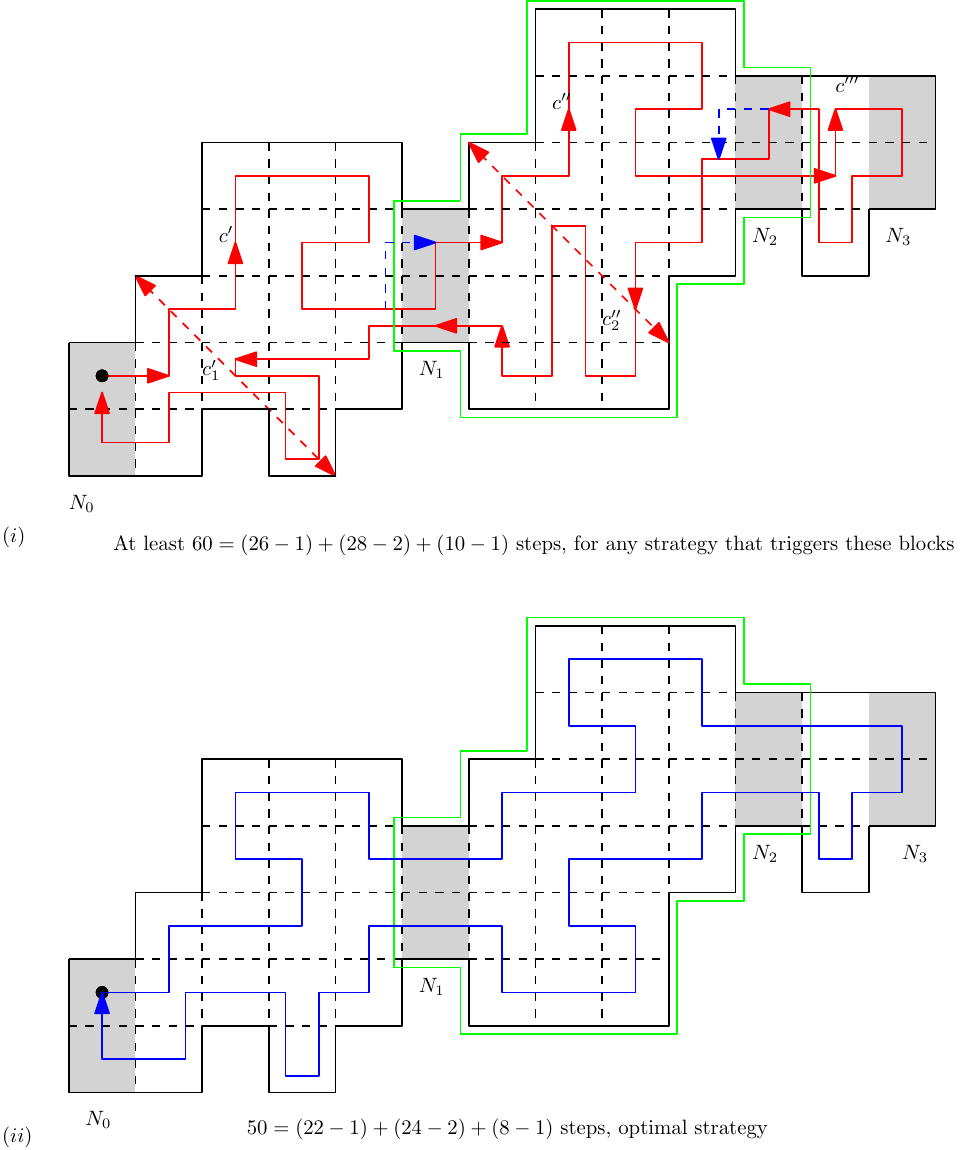}
\caption{The online strategy and the offline strategy profit from the merge in the same way. In inner blocks one can save two vertical steps, in outer blocks one can save one vertical step.}\label{fig-example-concat}
\end{center}
\end{figure}
\subsection{Concatenation and merge of blocks}
In a third step, we successively construct such blocks horizontally
and merge them together. On the right hand side of each single
(partially) presented block there is a new starting block (grey shaded block) for the next iteration to the right.
Note that with respect to the overall start and the single block construction, such a potentially new starting block of width~2 and height~1 on the
right hand side is always already fully known when it will be visited.  It is easy to see that we can start the next iteration in this ending grey block of the preceding construction. Figure \ref{fig-example-concat}~(i) shows an example of three such concatenated blocks, triggered by an online strategy.  The strategy starts in the known block $N_0$, at $c'$ the next starting block $N_1$ is already known, at $c''$ the next starting block $N_2$ is known and at $c'''$ the next (potentially) starting rectangle $N_3$ is given. Note that the first two blocks will be finally fully presented at the backward path after $c_2''$ and $c_1'$, respectively.

The merge of the blocks has some influence on the number of cells and steps. To analyse this behaviour let us first consider the corresponding optimal strategy; see for example Figure \ref{fig-example-concat}~(ii). The corresponding optimal strategy will pass all \emph{intermediate} merged starting and end rectangles optimally  (note that they are narrow passages) without any vertical movements. W.r.t. the above single block analysis we thus
will save one step for the first block on the right side. Two steps for the second block. Namely one vertical movement on the left side and one on the right side.  For the rightmost block we save one vertical movement  on its left side. The corresponding vertical movements will be omitted and this precisely fits to the number of cells that have been merged. In the given examples we have  $1+2+1=4$ cells less by merging the blocks when compared to the sum of the cells over the individual blocks.  In general for $n$ blocks we will save $1 + 2(n-2) +1=2(n-1)$ steps in total for the optimal strategy for any possible concatenation of the blocks.

How does an online strategy profit from the merge of the starting/ending rectangles?
For comparison to the overall online strategy we can argue as follows.  Since the narrow passages are all fully known in advance, 
an online strategy profits in the same way but not more.
 For example, for the analysis we can shift vertical movements into the corresponding neighbouring blocks without any further changes.  In total, the outer blocks can profit by one step less on the left or right side,
 respectively.  Each inner block can profit by two steps less which is in total again the same profit of $1 + 2(n-2) +1=2(n-1)$ steps less for
 $n$ successive arbitrary blocks.
\begin{figure}[htp]
\begin{center}
  \includegraphics[scale=0.58]{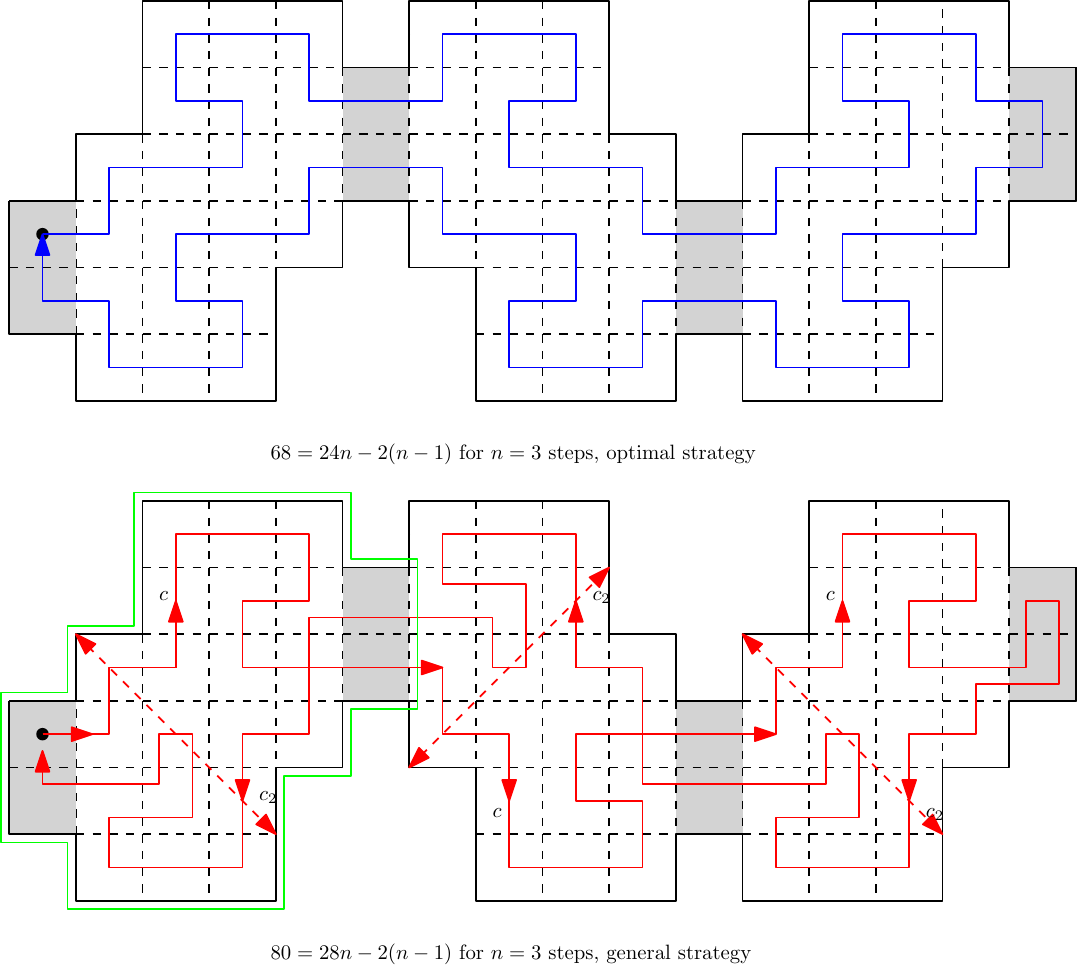}
\caption{For the smallest attainable ratio the online strategy has to trigger successive blocks of type (i). The bound $\frac{13}{11}$ is attained asymptotically.}
\label{fig-final-wc}\end{center}
\end{figure}

Finally, we argue that for the best ratio a strategy triggers block (i)
all the time, Figure~\ref{fig-final-wc} shows an example for $n=3$.
Note that the inner block is a mirrored variant for a different entry point.

\pagebreak
\begin{theorem}\label{theo-LB} Let ${ALG}$ be an online exploration strategy  and ${OPT}$ the corresponding optimal offline strategy. For any $\epsilon>0$, 
we can construct arbitrarily large simple grid polygons $P_{\epsilon}$ such that the ratio of $S_{ALG}(P_{\epsilon})$ against $S_{OPT}(P_{\epsilon})$ will be larger than~$\frac{13}{11}-\epsilon$. Therefore no strategy can attain a competitive ratio strictly smaller than~$\frac{13}{11}$.
\end{theorem}

\begin{proof}
We sum up by the previous argumentation. The adversary chooses an arbitrary $n\in \N$ for the number of blocks depending on $\epsilon$ (see the limit below). The strategy starts in a starting rectangle and depending on the movement, $n$~successive concatenated blocks will be presented. By its movement the strategy ${ALG}$ can actually successively \emph{choose} from the local ratios from Table~\ref{table-compare}  for single blocks. For the concatenation we have to take the common (strategy/optimum) gain of $2$ (inner blocks) and $1$ (outer blocks) into account. For a single ratio $\frac{S_A}{S_O}$ we consider $(S_A-2)$ and $(S_A-1)$ against $(S_O-2)$ and $(S_O-1)$, respectively. Obviously, the smallest possible ratio for any strategy will be attained, if the strategy triggers block (i) (or its mirrored analogue) all the time. After $n$ such blocks for the best attainable ratio a strategy ${ALG}$ will require at least $(28-1)+(28-2)(n-2)+(28-1)$ steps against $(24-1)+(24-2)(n-2)+(24-1)$ steps of an optimal strategy ${OPT}$ and
$$\lim_{n\to\infty} \frac{28 + 26(n-1)}{24 + 22(n-1)} = \frac{13}{11}$$ gives the first conclusion. 

To finally conclude w.r.t. inequality~(\ref{equ-compAn}) we show that there is no strategy $ALG$ that guarantees  $S_{ALG}(P) \leq C' \cdot S_{OPT}(P) + A$ for any $1\leq C'<\frac{13}{11}$ and any (fixed) additive constant $A$ and for any $P$. So let $C'=\frac{13}{11}-\delta$ for a fixed $\delta>0$. 
Thus the inequality would mean  $$\frac{S_{ALG}(P)}{ S_{OPT}(P)}\leq \frac{13}{11} +\left(\frac{A}{S_{OPT}(P)}-\delta\right) \mbox{ holds for any } P\,.$$ Now for $\delta$ and $A$ we can choose a first $n\in \N$ such that for any combination of the $n$ blocks and the corresponding $P_n$ we guarantee that $\frac{A}{S_{OPT}(P_n)}-\delta<0$ holds. %W.r.t. our block offer we can make use of the function $\frac{A}{8+6(n-1)}-\delta$ to find this $n\in \N$, i.e., the smallest block size (Block (b) and its optimal solution) that can appear is used for this calculation. 
Now for such an $n$ that depends on $C'$ and $A$ let $\epsilon:=\delta-\frac{A}{S_{OPT}(P_n)}$ with $\epsilon>0$.  
Now we can make use of the above limit process and find the first $n'\in \N$ such that $\frac{28 + 26(n'-1)}{24 + 22(n'-1)} > \frac{13}{11}-\epsilon$ holds. We make use of the maximum of $n$ and $n'$ and present an adversary strategy for this block size which will give the contradiction. 
\end{proof}

\section{Conclusion}\label{sect-concl}
We revisit the problem of exploring simple grid polygons online. In this setting, an agent has to visit every cell of an unknown grid environment without holes, starting from the boundary, and must finally return to the start. Information is limited to knowledge about the four cells adjacent to the current position and the agent has the ability to build a map of the detected cells. The performance is given by competitive analysis, i.e. the total number of cell visits is compared to an optimal offline solution where the environment is known in advance.
As we have also clarified that simple grid polygons and solid grid graphs are slightly different, the results can be translated.  

Here we show that the current best algorithm by Wei et al. \cite{wei} admits a flaw in the proof and consequently does not achieve the proposed competitive ratio of $\frac{7}{6}$. Therefore, the best known upper bound remains the $\frac{5}{4}$-competitive strategy by Kolenderska et al. \cite{kole}. Furthermore, we improve upon the previous lower bound: By using similar methods to \cite{kole}, we introduce a set of polygons where any online strategy fails to achieve a better competitive ratio than $\frac{13}{11}$ when compared to the optimal offline variant. The resulting gap of $\frac{3}{44}$ between the two bounds thus remains to be closed.

  \bibliographystyle{plain}
  \bibliography{references.bib}

\begin{thebibliography}{10}

\bibitem{arkin2009not}
Esther~M. Arkin, S{\'a}ndor~P. Fekete, Kamrul Islam, Henk Meijer, Joseph~S.B.
  Mitchell, Yurai N{\'u}{\~n}ez-Rodr{\'\i}guez, Valentin Polishchuk, David
  Rappaport, and Henry Xiao.
\newblock Not being (super) thin or solid is hard: A study of grid
  hamiltonicity.
\newblock {\em Computational Geometry}, 42(6-7):582--605, 2009.

\bibitem{feketetraveling}
S{\'a}ndor~P. Fekete, Christian Rieck, and Christian Scheffer.
\newblock On the traveling salesman problem in solid grid graphs.
\newblock In {\em {33rd European Workshop on Computational Geometry (EuroCG),
  Malmö, Sweden}}, February 2017.

\bibitem{gabriely2003competitive}
Yoav Gabriely and Elon Rimon.
\newblock Competitive on-line coverage of grid environments by a mobile robot.
\newblock {\em Computational Geometry}, 24(3):197--224, 2003.

\bibitem{icking2002competitive}
Christian Icking, Thomas Kamphans, Rolf Klein, and Elmar Langetepe.
\newblock On the competitive complexity of navigation tasks.
\newblock In {\em Sensor Based Intelligent Robots: International Workshop
  Dagstuhl Castle, Germany, October 15--20, 2000 Selected Revised Papers},
  pages 245--258. Springer, 2002.

\bibitem{ltepe}
Christian Icking, Tom Kamphans, Rolf Klein, and Elmar Langetepe.
\newblock Exploring simple grid polygons.
\newblock In {\em Proceedings of the 11th International Computing and
  Combinatorics Conference (COCOON’05), Vol. 3596 of Lecture Notes in
  Computer Science}, pages 524--533, 2005.

\bibitem{itai1982hamilton}
Alon Itai, Christos~H Papadimitriou, and Jayme~Luiz Szwarcfiter.
\newblock Hamilton paths in grid graphs.
\newblock {\em SIAM Journal on Computing}, 11(4):676--686, 1982.

\bibitem{kole}
Agnieszka Kolenderska, Adrian Kosowski, Micha{\l} Ma{\l}afiejski, and Pawe{\l}
  {\.Z}yli{\'n}ski.
\newblock An improved strategy for exploring a grid polygon.
\newblock In {\em Structural Information and Communication Complexity: 16th
  International Colloquium, SIROCCO 2009, Piran, Slovenia, May 25-27, 2009,
  Revised Selected Papers 16}, pages 222--236. Springer, 2010.

\bibitem{sleator1985amortized}
Daniel~D. Sleator and Robert~E. Tarjan.
\newblock Amortized efficiency of list update and paging rules.
\newblock {\em Communications of the ACM}, 28(2):202--208, 1985.

\bibitem{umans1997hamiltonian}
Christopher Umans and William Lenhart.
\newblock Hamiltonian cycles in solid grid graphs.
\newblock In {\em Proceedings 38th Annual Symposium on Foundations of Computer
  Science}, pages 496--505. IEEE, 1997.

\bibitem{wei}
Qi~Wei, Jie Sun, Xuehou Tan, Xiaolin Yao, and Yonggong Ren.
\newblock The simple grid polygon exploration problem.
\newblock {\em Journal of Combinatorial Optimization}, 41:625--639, 2021.

\end{thebibliography}


%\begin{thebibliography}{00}
%%% Text of bibliographic item
%
%\bibitem{}
%
%\end{thebibliography}
\end{document}